\documentclass[
  aps,
  prx,
  reprint,
  superscriptaddress,
  longbibliography,
  floatfix
]{revtex4-2}

\usepackage{amsmath,amssymb,amsfonts,amsthm}
\usepackage{graphicx}
\usepackage{textcomp}
\usepackage{xcolor}
\usepackage{siunitx}
\usepackage[dvipsnames]{xcolor}
\usepackage{hyperref}
\usepackage[percent]{overpic}
\definecolor{QNT_green}{HTML}{539E8E}

\hypersetup{
  colorlinks   = true, 
  urlcolor     = QNT_green, 
  linkcolor    = QNT_green, 
  citecolor   = QNT_green 
}

\newtheorem{theorem}{Theorem}[section]
\newtheorem{corollary}[theorem]{Corollary}
\newtheorem{lemma}[theorem]{Lemma}

\newtheorem{definition}[theorem]{Definition}

\begin{document}
\title{On Certifying Source Sampling Hardness in Quantum Generative Modeling}

\author{Chen-Yu Liu}
\email{chen-yu.liu@quantinuum.com}
\affiliation{Quantinuum, Partnership House, London, UK}

\author{Leonardo Placidi}
\email{leonardo.placidi@quantinuum.com}
\affiliation{Quantinuum, Otemachi Financial City Grand Cube, Tokyo, Japan}

\author{Enrico Rinaldi}
\email{enrico.rinaldi@quantinuum.com}
\affiliation{Quantinuum, Partnership House, London, UK}

\begin{abstract}
Quantum generative models are often motivated by circuit families whose output distributions are believed to be classically hard to sample from. When such models are trained on ordinary classical datasets, however, this hardness does not automatically transfer to the unknown data-generating distribution. We show that transferring sampling hardness via total-variation closeness from a quantum model to an unknown source requires certifying a global relation between the two distributions, thereby reducing the problem to distribution certification. Combining this reduction with existing certification lower bounds yields an exponential sample requirement for the high-entropy distributions relevant to many sampling-hardness proposals. Consequently, polynomially many samples cannot, in general, justify attributing sampling hardness to an unknown data-generating distribution. Moreover, even classically trivial distributions, such as the uniform distribution and product distributions, require exponentially many samples to certify in the absence of structural assumptions. Our results clarify the role of sampling hardness in quantum generative modeling and distinguish generator-level hardness from source-level hardness when learning from ordinary datasets.
\end{abstract}

\maketitle

\section{Introduction}

Quantum generative models (QGMs) are frequently motivated by the possibility of implementing probability distributions that are difficult to sample from using classical computers \cite{zhang2024generative, lloyd2018quantum, huang2025generative, recio2025train, coyle2020born, shen2026characterizing, aaronson2016complexity}. Prominent examples arise from quantum circuit families related to random circuit sampling \cite{arute2019quantum, morvan2024phase, google2025observation, bouland2019complexity, ZHU2022240, ransford202698, DeCross_2025}, instantaneous quantum polynomial-time (IQP) circuits \cite{ballo2026shallow, placidi2026impact, liu2026toward, bremner2011classical, leontica2024exploring}, and other sampling problems for which efficient classical simulation is believed to be impossible under standard complexity-theoretic assumptions \cite{bremner2011classical}. This observation suggests an appealing route toward quantum advantage in generative modeling: if a quantum device can efficiently generate samples from a distribution that no efficient classical algorithm can reproduce, then a generative model built from such a device \emph{may} possess capabilities unavailable to classical models. Classical sampling hardness has therefore become an important motivation for the design and analysis of QGMs.

At the same time, several proposals have demonstrated that QGMs based on such circuit families can be trained \cite{recio2025train, liu2026toward, kolarovszki2026generative, bako2025fermionic, liu2026generativeiqpcircuitlearning}, at least in restricted settings and to a prescribed accuracy \cite{herbst2025limits}, using objectives such as the maximum mean discrepancy (MMD) \cite{recio2025train, tuysuz2026quantum}. This naturally suggests the following line of reasoning. One begins with a quantum model family containing classically hard-to-sample distributions, trains a member of that family on a finite dataset, and then argues that the resulting model may inherit a quantum advantage from the sampling hardness of the underlying quantum generator. In this narrative, successful learning appears to connect complexity-theoretic hardness with the practical task of modeling data.

Recent work has begun examining whether this connection is justified. Herbst \emph{et al.}~\cite{herbst2025limits} study the interplay between anticoncentration, trainability, and classical surrogate sampling in QGMs. They show that output distributions exhibiting the flatness associated with standard sampling-hardness proposals are generally difficult to train, whereas sufficiently sparse distributions may become trainable but admit efficient classical surrogate samplers. Their work therefore asks whether the distributional properties underlying classical sampling hardness are compatible with practical trainability.

A complementary line of work concerns the statistical certification of quantum sampling devices. Hangleiter \emph{et al.}~\cite{hangleiter2019sample} consider the problem of certifying a \emph{known} target distribution from classical samples alone. Using instance-optimal identity-testing bounds, they show that non-interactive certification requires exponentially many samples for sufficiently flat target distributions, including those arising in prominent quantum-sampling proposals. Their result is information-theoretic: the target distribution is assumed to be completely specified, the certifier may have unlimited computational power, and the limitation arises purely from sample complexity.

The present work addresses a different question that lies between these two directions. In ordinary generative modeling, the data-generating distribution is typically \emph{unknown}; only a finite dataset is available. This raises a logically distinct problem: under what conditions can the classical sampling hardness of a specified quantum generator be transferred to the unknown source that generated the training data? We show that any such transfer based on total-variation closeness necessarily requires solving a known-target distribution-certification problem. Combining this reduction with the certification lower bounds of Hangleiter \emph{et al.} yields an exponential sample requirement for the high-entropy and almost-uniform distributions relevant to many sampling-hardness proposals. Consequently, polynomially many samples cannot, in general, justify attributing sampling hardness to an unknown data-generating distribution.

This observation changes the interpretation of sampling hardness in quantum generative modeling. Classical sampling hardness remains a meaningful complexity-theoretic property of specified quantum generators and promised source families. For ordinary datasets of unknown origin, however, the relevant question is not whether the latent source is classically hard to sample from, but whether a quantum model provides measurable advantages over strong classical baselines under the assumptions of the learning task.

Our contribution is to identify known-target certification as the missing statistical step in closeness-based transfer from generator-level sampling hardness to source-level hardness, and to derive its consequences for quantum generative modeling.

A related total-variation promise problem also appears in complexity theory. The \emph{Statistical Difference} problem is complete for the class Statistical Zero Knowledge (SZK)~\cite{sahai2003complete}. Although both settings formulate closeness using total-variation distance, they study different computational resources. Statistical Difference concerns the computational complexity of deciding closeness when descriptions of both samplers are given, whereas the present work concerns the information-theoretic sample complexity of certifying an unknown source from i.i.d. observations. Consequently, our results neither follow from nor establish SZK-hardness.

The remainder of this paper is organized as follows. In Sec.~\ref{sec:hardness_of_DG}, we distinguish generator-level sampling hardness from source-distribution sampling hardness, with particular emphasis on the difference between known and unknown data-generating mechanisms. In Sec.~\ref{sec:hardness-transfer}, we show that transferring sampling hardness from a specified quantum distribution to an unknown source requires certifying a sufficiently strong global relation between the two distributions. In Sec.~\ref{sec:certification-bounds}, we connect this requirement to known lower bounds for distribution certification and derive sample-complexity consequences for representative source families, including uniform, product, sparse, and Porter--Thomas-like distributions. In Sec.~\ref{sec:implications}, we discuss the implications for the interpretation of sampling-based quantum advantage in quantum generative modeling. Finally, Sec.~\ref{sec:conclusion} summarizes our conclusions.

\section{Hardness of Distributions and Generators}
\label{sec:hardness_of_DG}

The notion of classical sampling hardness can refer to several different objects, and these distinctions are essential in the context of quantum generative modeling. In particular, one should separate the hardness of a specified generator, the hardness of the unknown distribution underlying a dataset, the empirical fit of a trained model to finite samples, and the existence of a genuine quantum learning advantage. These statements are related, but none of them follows automatically from the others.

\subsection{Generator-level sampling hardness}
\label{sec:glsh}

Let \( \mathcal{G}_n=\{G_{\theta_n}:\theta_n\in\Theta_n\} \) be an ensemble of quantum generators acting on \(n\) qubits, and let \(Q_{\theta_n}\) denote the distribution obtained by measuring \(G_{\theta_n}\) in the computational basis.

We say that the ensemble exhibits \emph{generator-level classical sampling hardness} at total-variation tolerance \(\varepsilon(n)\) if there is no uniform classical randomized polynomial-time algorithm that, given a description of \(G_{\theta_n}\), produces samples from a distribution \(\widetilde Q_{\theta_n}\) satisfying
\begin{equation}
d_{\mathrm{TV}}
\left(
\widetilde Q_{\theta_n},
Q_{\theta_n}
\right)
\leq
\varepsilon(n)
\end{equation}
for the required fraction of instances \(\theta_n\) and for all sufficiently large \(n\). The required fraction of instances and the admissible scaling of \(\varepsilon(n)\) depend on the circuit family and on the precise average-case hardness conjecture used.

For the approximate-sampling hardness results considered here, this  conclusion is typically conditional on three ingredients. First, approximating the relevant output probabilities (or related counting quantities) of the circuit family to multiplicative precision is assumed to be \(\#\mathrm{P}\)-hard on average over a non-negligible fraction of instances \cite{bremner2011classical, bouland2019complexity, aaronson2011computational}. Second, the output distribution is assumed to anticoncentrate, so that these probabilities are not much smaller than their natural scale \(2^{-n}\) on a non-negligible fraction of instances.
Namely, there exist constants
\(a,\beta>0\), independent of \(n\), such that
\begin{equation}
\Pr_{\theta_n\sim\mu_n}
\left[
Q_{\theta_n}(x)
\geq
\frac{a}{2^n}
\right]
\geq
\beta
\end{equation}
for every fixed output string \(x\), where the probability is taken over the random choice of generator instance \(\theta_n\sim\mu_n\). And a second-moment formulation requires
\begin{equation}
\mathbb E_{\theta_n\sim\mu_n}
\left[
Q_{\theta_n}(x)^2
\right]
\leq
\frac{c_0}{2^{2n}}
\end{equation}
for some constant \(c_0>0\). Under suitable first-moment normalization, the second-moment condition implies a threshold anticoncentration bound through the Paley--Zygmund inequality \cite{petrov2007lower, Paley_Zygmund_1932}. Third, the polynomial hierarchy is assumed not to collapse. Under these assumptions, the existence of an efficient classical approximate sampler, combined with Stockmeyer's approximate-counting algorithm \cite{stockmeyer1983complexity}, would imply such a collapse.

Generator-level sampling hardness is therefore a statement about a specified circuit ensemble, an instance distribution \(\mu_n\), an approximation regime, an anticoncentration property, and an average-case complexity conjecture. It does not imply that every parameter choice in a variational model family is classically hard to sample from. Importantly, this definition does not concern whether the total-variation condition can be established from finite classical samples alone. In fact, later sections will show that these are fundamentally different questions.

\subsection{Source-distribution sampling hardness}
\label{sec:sdsh}

Source-distribution sampling hardness concerns the computational complexity of the process that generates the training data in a quantum generative modeling task. As with generator-level sampling hardness, it is defined for an asymptotic family of source distributions
\begin{equation}
\mathcal{P}_{\mathrm{data}}
=
\{P_{\mathrm{data},n}\}_{n\geq1},
\end{equation}
where each \(P_{\mathrm{data},n}\) is a distribution over \(\{0,1\}^n\).

\paragraph{Known-source distributions.}

Suppose that the data-generating mechanism is known. For example, the training data may be produced by a specified family of quantum circuits whose output distributions are believed to be classically hard to sample. In this setting, source-distribution sampling hardness is defined exactly as in Sec.~\ref{sec:glsh}. Namely, one asks whether there exists a uniform classical randomized polynomial-time algorithm that can approximately sample from the source distribution within the required total-variation accuracy.

While this is the same notion as generator-level sampling hardness, the only difference is the object under consideration: generator-level hardness concerns the output distribution of a specified quantum model, whereas source-distribution hardness concerns the distribution that generated the training data when the source itself is known, for example through an explicit circuit description, a physical preparation procedure, or a promised source family. An example can be found in \cite{huang2025generative}.

\paragraph{Unknown-source distributions.}

In ordinary generative modeling, however, the data-generating process is typically not known or not given. Instead, the learner is simply given a finite dataset
\begin{equation}
D=\{x_1,\ldots,x_m\},
\qquad
x_i
\overset{\mathrm{i.i.d.}}{\sim}
P_{\mathrm{data}},
\end{equation}
where the underlying source distribution is latent.

This seemingly small change fundamentally alters the problem. Unlike the known-source setting above, the source distribution is no longer specified by a circuit description, probability oracle, or trusted physical preparation procedure. The only available evidence consists of finitely many observed samples.

Consequently, in this case, although generator-level hardness and source-distribution hardness are formally defined by the same approximation criterion, they have fundamentally different operational meaning. In the former, the distribution is explicitly specified by a known computational object. In the latter, the distribution is latent, and any complexity-theoretic claim about it must be justified using only finite observations or additional assumptions about the source. 

Throughout the remainder of this paper we write $P_{\mathrm{data}}$ whenever the dependence on $n$ is clear from context. The central question of this work concerns precisely this unknown-source setting: whether the classical sampling hardness of an unknown data-generating distribution can be meaningfully inferred or certified from the finite datasets encountered in ordinary quantum generative modeling.

\subsection{Finite-sample fit}

Let \(Q_{\widehat\theta}\) be the output distribution of a trained QGM. A statement of finite-sample fit means that \(Q_{\widehat\theta}\) agrees with the dataset according to a chosen empirical objective, such as empirical likelihood, an MMD loss, a kernel statistic, or performance on a held-out sample. Symbolically, one may write
\begin{equation}
\widehat{\mathcal L}_D(Q_{\widehat\theta})
\leq \delta_{\mathcal L},
\end{equation}
for some empirical loss \(\widehat{\mathcal L}_D\) and tolerance \(\delta_{\mathcal L}\).
Such a statement is statistical rather than complexity-theoretic. It shows that the trained model is compatible with the observed data under the selected criterion. It does not establish that
\begin{equation}
Q_{\widehat\theta}=P_{\mathrm{data}},
\end{equation}
nor that \(Q_{\widehat\theta}\) is classically hard to sample from, nor that \(P_{\mathrm{data}}\) itself is classically hard. 
In particular, agreement on finitely many samples may leave the model's behavior on the unobserved region determined primarily by, where applicable, the model architecture, inductive bias, and training procedure.

\section{Hardness transfer requires certification}
\label{sec:hardness-transfer}

Suppose that a quantum generator produces a distribution \(Q_n\) that is classically hard to sample within some total-variation tolerance. An interesting question is ``Under what conditions can this hardness be attributed to another distribution \(P_n\), such as the unknown distribution underlying a training dataset?''

The relevant condition is global distributional closeness. If \(P_n\) is sufficiently close to \(Q_n\) in total variation distance, then any efficient classical sampler for \(P_n\) would also provide an efficient approximate sampler for \(Q_n\). This observation gives the following elementary hardness-transfer result.

\begin{lemma}[Hardness transfer under total-variation closeness]
\label{lem:hardness-transfer}
Let
\begin{equation}
\mathcal{Q}
=
\{Q_n\}_{n\geq 1}
\end{equation}
be a distribution family that is classically hard to sample within total-variation error \(\varepsilon(n)>0\). That is, under the complexity-theoretic assumptions associated with the hardness result, there is no uniform classical randomized polynomial-time algorithm whose output distribution \(\widetilde Q_n\) satisfies
\begin{equation}
d_{\mathrm{TV}}
\left(
\widetilde Q_n,Q_n
\right)
\leq
\varepsilon(n)
\end{equation}
for all sufficiently large \(n\).

Let
\begin{equation}
\mathcal{P}
=
\{P_n\}_{n\geq 1}
\end{equation}
be another distribution family satisfying
\begin{equation}
d_{\mathrm{TV}}
\left(
P_n,Q_n
\right)
\leq
\delta(n),
\qquad
0\leq \delta(n)<\varepsilon(n).
\end{equation}
Then \(\mathcal{P}\) is classically hard to sample within total-variation error
\begin{equation}
\varepsilon_{\mathcal P}(n)
=
\varepsilon(n)-\delta(n).
\end{equation}
\end{lemma}

\begin{proof}
Assume, for contradiction, that there exists a uniform classical randomized polynomial-time algorithm whose output distribution \(\widetilde P_n\) satisfies
\begin{equation}
d_{\mathrm{TV}}
\left(
\widetilde P_n,P_n
\right)
\leq
\varepsilon(n)-\delta(n)
\end{equation}
for all sufficiently large \(n\). By the triangle inequality,
\begin{align}
d_{\mathrm{TV}}
\left(
\widetilde P_n,Q_n
\right)
&\leq
d_{\mathrm{TV}}
\left(
\widetilde P_n,P_n
\right)
+
d_{\mathrm{TV}}
\left(
P_n,Q_n
\right)
\nonumber\\
&\leq
\varepsilon(n)-\delta(n)+\delta(n)
\nonumber\\
&=
\varepsilon(n).
\end{align}
Thus the same algorithm would approximately sample from \(Q_n\) within the tolerance ruled out by the assumed hardness of \(\mathcal Q\), a contradiction.
\end{proof}

Lemma~\ref{lem:hardness-transfer} identifies the condition required to transfer sampling hardness from a specified quantum distribution to a data-generating distribution. It is not sufficient that \(Q_n\) belong to a hard quantum model family, nor that a trained model achieve a small empirical loss on samples from \(P_n\). One must establish a global relation of the form
\begin{equation}
d_{\mathrm{TV}}
\left(
P_n,Q_n
\right)
\leq
\delta(n)
<
\varepsilon(n),
\end{equation}
where \(\varepsilon(n)\) is the approximation tolerance appearing in the generator-level hardness statement. This observation separates two logically distinct questions. The first is complexity-theoretic:

\begin{quote}
Given a specified distribution family \(Q_n\), is there an efficient classical algorithm that samples from it within error \(\varepsilon(n)\)?
\end{quote}

The second is statistical:

\begin{quote}
Given only samples from an unknown source \(P_n\), can one certify that \(P_n\) lies within distance \(\delta(n)\) of \(Q_n\)?
\end{quote}

The first question concerns the hardness of a known computational object. The second concerns the certification of a relation between a known target distribution and a latent source distribution. The generator-level hardness result does not answer the second question.

This distinction is especially important in quantum generative modeling. Suppose that a trained quantum model produces a distribution \(Q_{\widehat\theta,n}\) that is believed to be classically hard to sample. To infer that the unknown data-generating distribution \(P_{\mathrm{data},n}\) shares this hardness, one would need to establish
\begin{equation}
d_{\mathrm{TV}}
\left(
P_{\mathrm{data},n},
Q_{\widehat\theta,n}
\right)
\leq
\delta(n)
\end{equation}
for a tolerance smaller than the hardness threshold of the trained quantum distribution. Finite-sample agreement under a training objective does not, by itself, provide such a certificate.

We therefore define the certification problem that mediates any closeness-based transfer of sampling hardness.

\begin{definition}[Known-target certification complexity]
\label{def:cert-complexity}
Let \(Q\) be a known distribution over a finite domain \(\mathcal X\), and let \(0\leq \varepsilon_0<\varepsilon_1\leq 1\). The certification sample complexity
\begin{equation}
m_{\mathrm{cert}}
\left(
Q;\varepsilon_0,\varepsilon_1
\right)
\end{equation}
is the minimum number of independent samples from an unknown distribution \(P\) required by any test that distinguishes
\begin{equation}
d_{\mathrm{TV}}(P,Q)\leq\varepsilon_0
\end{equation}
from
\begin{equation}
d_{\mathrm{TV}}(P,Q)\geq\varepsilon_1.
\end{equation}
\end{definition}

The identity-testing problem corresponds to the special case \(\varepsilon_0=0\). More generally, the separated thresholds \(\varepsilon_0<\varepsilon_1\) describe tolerant certification.
Any $(\varepsilon_0,\varepsilon_1)$-tolerant certification procedure also solves known-target identity testing at separation  $\varepsilon_1$, since the identity case $P=Q$ satisfies  $d_{\mathrm{TV}}(P,Q)=0\leq\varepsilon_0$. 
For the purpose of hardness transfer, the relevant certification accuracy is set by the generator-level hardness tolerance. 
If \(Q_n\) is hard to sample within error \(\varepsilon(n)\), then a sound hardness-transfer argument must certify closeness at some threshold \(\delta(n)<\varepsilon(n)\). Consequently, the number of samples required to justify the transfer is lower-bounded by the corresponding certification complexity of \(Q_n\). Throughout, certification complexity refers to success probability at least \(2/3\), unless stated otherwise.

\begin{corollary}[Certification requirement for hardness transfer]
\label{cor:certification-requirement}
Let \(Q_n\) be classically hard to sample within total-variation error \(\varepsilon(n)\). Any argument that transfers this hardness to an unknown distribution \(P_n\) solely by establishing
\begin{equation}
d_{\mathrm{TV}}(P_n,Q_n)\leq\delta(n),
\qquad
\delta(n)<\varepsilon(n),
\end{equation}
must solve the corresponding distribution-certification problem. Therefore, its sample requirement is at least
\begin{equation}
m_{\mathrm{cert}}
\left(
Q_n;
\delta(n),
\varepsilon_1(n)
\right)
\end{equation}
for an alternative threshold \(\varepsilon_1(n)>\delta(n)\) chosen to make the certificate statistically sound.
\end{corollary}

Corollary~\ref{cor:certification-requirement} does not yet determine the magnitude of the sample complexity. It reduces the source-hardness question to a standard distribution-certification problem. In the next section, we use known lower bounds for this problem to quantify the number of samples required for the high-entropy and almost-uniform distributions relevant to quantum sampling-hardness arguments.

\section{Certification lower bounds}
\label{sec:certification-bounds}

We now quantify the sample requirement identified in Sec.~\ref{sec:hardness-transfer}. Recall that a transfer of sampling hardness from a known hard distribution \(Q\) to an unknown source distribution \(P\) requires a global closeness statement between the two. We therefore consider the most favorable sample-only certification setting as in \cite{hangleiter2019sample}: the target distribution \(Q\) is known completely, the certifier has unlimited computational power, and only the number of samples drawn from \(P\) is counted.

Even in this favorable setting, the required number of samples can be exponential in the problem size. The relevant lower bound follows from the optimal identity-testing result of Valiant and Valiant \cite{valiant2017automatic}, as applied to quantum-sampling certification by Hangleiter \emph{et al.}\cite{hangleiter2019sample}.

\subsection{Known-target identity testing}

Let \(Q=(q_x)_{x\in\mathcal X}\) be a known distribution on a finite sample space \(\mathcal X\). For \(\gamma>0\), let \(Q_{-\gamma}^{-\max}\) denote the subnormalized vector obtained by removing the largest probability of \(Q\) and removing a collection of its smallest probabilities whose total weight is at most \(\gamma\). For a nonnegative vector \(v\), define
\begin{equation}
\|v\|_{2/3}
=
\left(
\sum_x v_x^{2/3}
\right)^{3/2}.
\end{equation}

The following statement is a reformulation of the identity-testing lower bound used in Theorem 2 in Ref.~\cite{hangleiter2019sample}.

\begin{theorem}[Known-target certification lower bound]
\label{thm:known-target-certification}
Let \(Q\) be a known target distribution, and suppose that a test receives \(m\) independent samples from an unknown distribution \(P\). Any test that distinguishes
\(
P=Q
\)
from
\(
\|P-Q\|_1>\rho
\)
with constant completeness and soundness error requires
\begin{equation}
m
\geq
c\,
\max
\left\{
\frac{1}{\rho},
\frac{1}{\rho^2}
\left\|
Q_{-2\rho}^{-\max}
\right\|_{2/3}
\right\},
\end{equation}
where \(c>0\) is a universal constant.
\end{theorem}

Theorem~\ref{thm:known-target-certification} concerns identity testing: the accepting case is exactly \(P=Q\). This is weaker than tolerant certification, in which the test must also accept all distributions in a nonzero neighborhood of \(Q\). Consequently, the theorem already gives a lower bound for any more demanding procedure intended to certify global closeness between an unknown source and a known target.

Since
\begin{equation}
\|P-Q\|_1
=
2d_{\mathrm{TV}}(P,Q),
\end{equation}
with $m_{\mathrm{cert}}(Q,\tau):=m_{\mathrm{cert}}(Q;0,\tau)$.
Applying Theorem~\ref{thm:known-target-certification} with the \(\ell_1\)-separation \(\rho=2\tau\) gives
\begin{equation}
\label{eq:m_cert_lb}
m_{\mathrm{cert}}(Q,\tau)
\in
\Omega
\left[
\max
\left\{
\frac{1}{\tau},
\frac{1}{\tau^2}
\left\|
Q_{-4\tau}^{-\max}
\right\|_{2/3}
\right\}
\right],
\end{equation}
The central quantity is thus not simply the cardinality of the sample space. It is the truncated \(\ell_{2/3}\) quasi-norm of the target distribution, which measures how broadly its non-negligible probability mass is spread.

\subsection{A min-entropy consequence}

For a distribution \(Q\), let
\begin{equation}
H_\infty(Q)
=
-\log_2 \max_x Q(x)
\end{equation}
denote its min-entropy. Writing
\begin{equation}
h=H_\infty(Q),
\qquad
q_{\max}=2^{-h},
\end{equation}
the bounds of Ref.~\cite{hangleiter2019sample} (more in Appendix~\ref{app:certification}) imply
\begin{equation}
\left\|
Q_{-\gamma}^{-\max}
\right\|_{2/3}
\geq
2^{h/2}
\left(
1-\gamma-2^{-h}
\right)^{3/2},
\end{equation}
whenever the expression in parentheses is positive.

Combining this inequality with Theorem~\ref{thm:known-target-certification} gives the following consequence.

\begin{corollary}[Certification cost from min-entropy]
\label{cor:min-entropy-certification}
Let \(Q\) be a known target distribution with
\begin{equation}
H_\infty(Q)=h.
\end{equation}
For a fixed total-variation threshold
\(0<\tau<1/4\), and $\gamma = 4 \tau$, any sample-only identity test for \(Q\) requires
\begin{equation}
m_{\mathrm{cert}}(Q,\tau)
\in
\Omega
\left[
\frac{2^{h/2}}{\tau^2}
\left(
1-4\tau-2^{-h}
\right)^{3/2}
\right].
\end{equation}
In particular, for constant \(\tau\) and \(h\to\infty\),
\begin{equation}
m_{\mathrm{cert}}(Q,\tau)
\in
\Omega\left(2^{h/2}\right).
\end{equation}
\end{corollary}

Thus, if
\begin{equation}
H_\infty(Q_n)\in\Omega(n),
\end{equation}
then known-target certification requires exponentially many samples:
\begin{equation}
m_{\mathrm{cert}}(Q_n,\tau)
\in
2^{\Omega(n)}
\end{equation}
for any fixed sufficiently small \(\tau\).

The min-entropy bound is convenient but need not be tight for every distribution. Two targets with the same min-entropy may have different truncated \(\ell_{2/3}\) quasi-norms and hence different identity-testing complexities. Whenever possible, the sharper quantity in Theorem~\ref{thm:known-target-certification} should therefore be used.

\subsection{Consequence for source-hardness transfer}

We now combine Corollary~\ref{cor:min-entropy-certification} with the hardness-transfer result of Sec.~\ref{sec:hardness-transfer}.

\begin{corollary}[Sample requirement for source-hardness transfer]
\label{cor:source-hardness-transfer}
Let \(\{Q_n\}_{n\geq1}\) be a known distribution family that is
classically hard to sample within total-variation error \(\varepsilon(n)\).
Suppose that a source-hardness claim for an unknown family
\(\{P_n\}_{n\geq1}\) is based solely on certifying
\[
d_{\mathrm{TV}}(P_n,Q_n)\leq \delta(n)<\varepsilon(n).
\]
Let \(\varepsilon_1(n)>\delta(n)\) be the rejection threshold of the
corresponding tolerant certification problem. If
\[
H_\infty(Q_n)=h_n,
\]
then, for constant certification error probability, this route requires
\[
m
\in
\Omega\!\left[
\frac{2^{h_n/2}}{\varepsilon_1(n)^2}
\left(
1-4\varepsilon_1(n)-2^{-h_n}
\right)^{3/2}
\right],
\]
whenever the factor in parentheses is positive.
In particular, if \(h_n\in\Omega(n)\)
and $0<\varepsilon_1(n)\le c<1/4$ for some constant $c$, then the required
number of samples is exponential in \(n\).
\end{corollary}

Corollary~\ref{cor:source-hardness-transfer} is the direct implication for quantum generative modeling. Even when the candidate hard target \(Q_n\) is known completely, a polynomial-size dataset cannot generally certify the global closeness required to transfer its sampling hardness to an unknown data source in the high-min-entropy regime. The ordinary dataset-first setting provides no more information than this known-target problem and generally provides less.

\subsection{Representative distribution regimes}
\label{sec:certification-examples}

We next illustrate the certification bound in several representative regimes. These examples separate three notions that should not be conflated: entropy, classical sampling complexity, and learnability under structural assumptions.

\subsubsection{Uniform distribution}

Let \(U_n\) denote the uniform distribution over
\(\{0,1\}^n\):
\begin{equation}
U_n(x)=2^{-n}.
\end{equation}
Its min-entropy is
\begin{equation}
H_\infty(U_n)=n.
\end{equation}
Therefore, for constant \(\tau\),
\begin{equation}
m_{\mathrm{cert}}(U_n,\tau)
\in
\Omega\left(2^{n/2}\right).
\end{equation}

This example demonstrates that large certification complexity does not imply classical sampling hardness. The uniform distribution is classically trivial to sample, yet certifying it in total variation against arbitrary alternative distributions requires exponentially many samples. The obstruction arises from flatness on an exponentially large domain, not from computational hardness.

For the uniform distribution, the general identity-testing bound is also tight up to constant and accuracy-dependent factors \cite{Paninski}:
\begin{equation}
m_{\mathrm{cert}}(U_n,\tau)
\in
\Theta\left(\frac{2^{n/2}}{\tau^2}\right)
\end{equation}
in the usual constant-confidence setting.

\subsubsection{Porter--Thomas-like distributions}

Consider a target distribution on \(N=2^n\) outcomes whose largest probability satisfies
\begin{equation}
q_{\max}
\in
O\left(\frac{\operatorname{poly}(n)}{2^n}\right).
\end{equation}
Equivalently,
\begin{equation}
H_\infty(Q_n)
\geq
n-O(\log n).
\end{equation}
This includes the typical scaling expected for Porter--Thomas-like
output distributions. Suppose, more explicitly, that
\begin{equation}
q_{\max}
=
\max_x Q_n(x)
\leq
\frac{Cn}{2^n}
\end{equation}
for some constant \(C>0\). Then
\begin{equation}
H_\infty(Q_n)
\geq
n-\log_2 n-\log_2 C.
\end{equation}
Applying Corollary~\ref{cor:min-entropy-certification} yields the explicit lower bound with a constant \(c_{\mathrm{cert}}>0\) that
\begin{equation}
m_{\mathrm{cert}}(Q_n,\tau)
\geq
\frac{c_{\mathrm{cert}}}{\tau^2}
\frac{2^{n/2}}{\sqrt{Cn}}
\left(
1-4\tau-\frac{Cn}{2^n}
\right)^{3/2}.
\end{equation}
Consequently, for fixed \(0<\tau<1/4\),
\begin{equation}
m_{\mathrm{cert}}(Q_n,\tau)
\in
\Omega\!\left(
\frac{2^{n/2}}{\sqrt n}
\right).
\end{equation}

Thus, the almost-uniform output distributions commonly associated with random quantum circuits remain exponentially expensive to certify from classical samples.

\subsubsection{Sparse distributions}

Suppose \(Q_n\) is uniform on a support of size \(s_n\):
\begin{equation}
Q_n(x)
=
\begin{cases}
s_n^{-1}, & x\in S_n,\\
0, & x\notin S_n.
\end{cases}
\end{equation}
Then
\begin{equation}
H_\infty(Q_n)=\log_2 s_n.
\end{equation}
Applying Corollary~\ref{cor:min-entropy-certification}
gives the lower bound
\begin{equation}
m_{\mathrm{cert}}(Q_n,\tau)
\in
\Omega
\left(
\frac{\sqrt{s_n}}{\tau^2}
\right)
\end{equation}
For distributions that are uniform over their support, this lower bound is known to be tight \cite{Paninski,valiant2017automatic}, so that
\begin{equation}
m_{\mathrm{cert}}(Q_n,\tau)
\in
\Theta
\left(
\frac{\sqrt{s_n}}{\tau^2}
\right)
\end{equation}
for constant confidence. If
\begin{equation}
s_n=2^{\alpha n}
\end{equation}
for some constant \(0<\alpha\leq1\), then
\begin{equation}
m_{\mathrm{cert}}(Q_n,\tau)
\in
\Theta\!\left(
\frac{2^{\alpha n/2}}{\tau^2}
\right).
\end{equation}

This example exhibits a continuum between polynomially supported sources and distributions spread across the full exponentially large domain. The relevant distinction is not simply ``sparse'' versus ``dense'', but the effective number of outcomes carrying non-negligible probability mass.

\subsubsection{Product distributions and the role of structural promises}

Finally, consider an \(n\)-bit product distribution
\begin{equation}
Q_n(x)
=
\prod_{j=1}^n
p_j^{x_j}(1-p_j)^{1-x_j}.
\end{equation}
When all parameters are bounded away from zero and one, the min-entropy is extensive:
\begin{equation}
H_\infty(Q_n)
=
-\sum_{j=1}^n
\log_2\max\{p_j,1-p_j\}
\in
\Theta(n).
\end{equation}
The unrestricted known-target identity-testing problem, in which the unknown alternative \(P\) may be any distribution over \(\{0,1\}^n\), therefore has exponential sample complexity according to Corollary~\ref{cor:min-entropy-certification}.

This does not mean that product distributions are exponentially hard to learn. If the learner is promised that the unknown source is itself a product distribution, the source is determined by the \(n\) one-bit marginals. These parameters can be estimated from polynomially many samples, and the resulting structured learning problem can be sample-efficient.

This example is important for the interpretation of our result. High min-entropy alone does not imply computational hardness, poor learnability, or the absence of useful inductive structure. Rather, it implies a large sample requirement for distribution-free global certification against arbitrary alternatives. Structural promises can reduce the statistical problem, but those promises constitute additional information about the source and cannot be inferred from the unstructured dataset alone.

\begin{table*}[t]
\centering
\caption{
The role of source information in quantum generative modeling. Classical sampling complexity is a property of a specified distribution or generative mechanism, whereas sample-only certification is a statistical task involving an unknown source. When the source is known, for example through a circuit description or trusted preparation procedure, its sampling complexity can be analyzed directly and no sample-only identity certificate is required. When the source is unknown and only i.i.d.\ samples are available, certifying that it equals or is close to a specified target may require exponentially many samples, even when the target distribution is classically trivial to sample.
}
\label{tab:source_information}
\begin{tabular}{lccc}
\hline
\textbf{Target distribution} & \textbf{Source information} & \textbf{Classical sampling} & \textbf{Certification from samples alone}
\\
\hline
Uniform & Known & Easy & Not applicable
\\
Uniform & Unknown; samples only & Easy & $\Theta(2^{n/2})$
\\
Product & Known & Easy &  Not applicable
\\
Product & Unknown; samples only & Easy & Exponential in the unrestricted setting
\\
Random-circuit sampling & Known & Hard under standard conjectures & Not applicable
\\
Random-circuit sampling & Unknown; samples only & Not established from the samples & Exponential for sufficiently flat targets
\\
\hline
\end{tabular}
\end{table*}

\section{Implications for quantum generative modeling}
\label{sec:implications}

The preceding sections distinguish two conceptually different settings. In the first, the data-generating mechanism is specified explicitly, for example by a known quantum circuit family or a trusted physical preparation procedure. In this case, classical sampling hardness is a well-defined complexity-theoretic property of the source itself. In the second, which is the setting encountered in most practical generative modeling tasks, the source distribution is unknown and only a finite dataset is available. Our results concern this latter setting.

The key observation is that generator-level sampling hardness and source-distribution sampling hardness are connected only through a global relation between the two distributions. In particular, Lemma~\ref{lem:hardness-transfer} shows that transferring hardness from a quantum generator to an unknown source requires a certificate that the two distributions are sufficiently close in total variation distance. Section~\ref{sec:certification-bounds} then shows that establishing such a certificate from samples alone may require exponentially many observations for the high-entropy distributions relevant to many sampling-hardness proposals.

Consequently, for an ordinary dataset, successful training of a quantum generative model should not by itself be interpreted as evidence that the underlying data-generating distribution is classically hard to sample from. Finite-sample agreement establishes only that the trained model reproduces the observed data according to the chosen training objective. Without additional information about the source, such agreement does not justify transferring complexity-theoretic hardness from the model family to the unknown data-generating process.

A second, independent issue concerns the relation between generator-level hardness and the trained quantum model itself. The classical sampling hardness results discussed in Sec.~\ref{sec:glsh} are typically established for random circuit instances drawn from an ensemble or for a non-negligible fraction of instances under an associated average-case hardness conjecture. By contrast, the parameters of a trained QGM are not sampled from this ensemble. They are selected adaptively by an optimization procedure driven by the training data. Thus, the generator-level hardness of the underlying model family does not, by itself, imply that the particular trained instance remains in the hard-to-simulate regime. This issue is complementary to recent work examining whether the distributional properties associated with sampling hardness are compatible with trainability~\cite{herbst2025limits}. Here we assume, for the sake of argument, that a trained model does retain generator-level hardness, and ask the subsequent question of whether this hardness can be attributed to the unknown data-generating distribution.

The examples in Table~\ref{tab:source_information} illustrate that sample-only certification complexity is fundamentally different from computational sampling complexity. Even distributions that are classically trivial to sample, such as the uniform distribution or product distributions, may require exponentially many samples to certify in the absence of structural assumptions. Conversely, specifying the source family or generation mechanism fundamentally changes the statistical problem. \textit{The distinction is therefore not between quantum and classical models, but between specified and unspecified sources.}

This distinction clarifies how sampling-hardness arguments should be interpreted across the existing QGM literature. In promised-source settings, where the data are generated by a specified quantum process or a trusted family of quantum circuits, the sampling hardness of the source is part of the problem formulation rather than something inferred from the dataset. Examples include recent proposals for generative  quantum advantage~\cite{huang2025generative}, where the source distribution is explicitly assumed to originate from a classically hard-to-simulate quantum sampler. Our results do not challenge such settings, they concern the ordinary dataset-first scenario, in which only classical samples are provided and the data-generating  mechanism is unknown.  

More generally, works motivating Born machines and IQP-based generative models through the existence of classically hard distributions within their model families~\cite{coyle2020born, ballo2026shallow, placidi2026impact} should be interpreted with care when applied to ordinary datasets. The presence of hard-to-sample distributions within a generator family is an important complexity-theoretic property of the model, but it does not by itself establish that a trained instance  retains this hardness, nor that the unknown data-generating distribution inherits it. The latter requires an additional  distributional certificate, whose sample complexity is the subject of the present work.

These observations suggest that source-level sampling hardness is, in general, not an operationally accessible quantity for ordinary datasets of unknown origin. Consequently, practical claims of quantum generative utility should instead be supported by quantities that are operationally observable or empirically verifiable, such as predictive performance on held-out data, faithful reproduction of task-relevant observables, robustness, or resource-adjusted comparison against strong classical baselines under comparable assumptions.

\section{Conclusion}
\label{sec:conclusion}

We have distinguished the classical sampling hardness of a specified quantum generator from the hardness of an unknown data-generating distribution in quantum generative modeling. While the former is a well-defined complexity-theoretic property of a specified generative mechanism, the latter is a property of a latent source observed only through finite samples. Transferring sampling hardness from a quantum model to an unknown source therefore requires certifying a global relation between the two distributions.

By connecting this observation to existing certification lower bounds, we showed that such a certificate may require exponentially many samples for the high-entropy and almost-uniform distributions relevant to many sampling-hardness proposals. This limitation is not unique to classically hard distributions, even distributions that are trivial to sample, such as the uniform distribution or product distributions, may require exponentially many samples to certify in the absence of additional structural assumptions.

Our results therefore suggest that, for ordinary datasets of unknown origin, sampling hardness should not be regarded as an operational basis for claims of quantum generative advantage. Instead, meaningful advantage claims require additional information about the source, or should be established through task-dependent performance and comparison with strong classical baselines.

\begin{acknowledgments}
We are grateful to Matthew DeCross, Robin Lorenz, Natansh Mathur, and Harry Buhrman for helpful discussions.
\end{acknowledgments}

\bibliography{refs}

\clearpage
\onecolumngrid
\appendix

\section{Certification lower bounds and the min-entropy consequence}
\label{app:certification}

For completeness, we summarize the identity-testing result underlying
Sec.~\ref{sec:certification-bounds} and derive the min-entropy
corollary used throughout this work.

\subsection{Known-target identity testing}

The following theorem is a restatement of the
known-target identity-testing lower bound established by
Valiant and Valiant~\cite{valiant2017automatic}, as
reformulated by Hangleiter \emph{et al.}~\cite{hangleiter2019sample}
for the certification of quantum sampling devices.

\begin{theorem}[Hangleiter {\it et al.}]
\label{thm:hangleiter_appendix}
Let \(Q\) be a known probability distribution over a finite
sample space \(\mathcal X\), and let
\(Q_{-\gamma}^{-\max}\) denote the distribution obtained by
removing the largest probability together with the smallest
probabilities whose total mass is at most \(\gamma\).
Then any algorithm that distinguishes $P=Q$
from $\|P-Q\|_1>\rho$
with constant completeness and soundness error requires
\begin{equation}
m
\ge
c
\max
\left\{
\frac1\rho,
\frac1{\rho^2}
\left\|
Q_{-2\rho}^{-\max}
\right\|_{2/3}
\right\},
\end{equation}
where
\begin{equation}
\|v\|_{2/3}
=
\left(
\sum_i v_i^{2/3}
\right)^{3/2},  
\end{equation}
and \(c>0\) is a universal constant.
\end{theorem}

The proof is not reproduced here, since it follows directly
from the instance-optimal identity-testing framework of
Valiant and Valiant~\cite{valiant2017automatic}; see
Theorem~2 of Ref.~\cite{hangleiter2019sample}.
The remainder of this appendix derives the min-entropy
bound used in Sec.~\ref{sec:certification-bounds}.


\subsection{A lower bound from min-entropy}

Let
$
h
=
H_\infty(Q)
=
-\log_2
\max_x Q(x),
$
and denote
$
q_{\max}
=
2^{-h}.
$
After removing the largest probability together with
probability mass at most \(\gamma\),
the remaining vector
$
Q_{-\gamma}^{-\max}
=
(q_i)_i
$
has total probability mass
\begin{equation}
M
=
\sum_i q_i
\ge
1-\gamma-q_{\max}
=
1-\gamma-2^{-h}.
\end{equation}
Moreover, every remaining component satisfies
\begin{equation}
q_i
\le
q_{\max}
=
2^{-h}.
\end{equation}
Therefore,

\begin{equation}
q_i^{2/3}
=
q_i\,q_i^{-1/3}
\ge
q_i\,2^{h/3},
\end{equation}
where the inequality follows from
\(q_i^{-1/3}\ge2^{h/3}\).
Summing over all remaining entries gives
\begin{equation}
\sum_i
q_i^{2/3}
\ge
2^{h/3}
\sum_i q_i
=
2^{h/3}M.
\end{equation}
Applying the definition of the
\(\ell_{2/3}\) quasi-norm,

\begin{align}
\left\|
Q_{-\gamma}^{-\max}
\right\|_{2/3}
&=
\left(
\sum_i
q_i^{2/3}
\right)^{3/2}
\nonumber\\
&\ge
\left(
2^{h/3}M
\right)^{3/2}
\nonumber\\
&=
2^{h/2}
M^{3/2}
\nonumber\\
&\ge
2^{h/2}
\left(
1-\gamma-2^{-h}
\right)^{3/2}.
\label{eq:minentropy_appendix}
\end{align}
This establishes the lower bound quoted in
Sec.~\ref{sec:certification-bounds}.


\subsection{Derivation of Corollary~\ref{cor:min-entropy-certification}}
Theorem~\ref{thm:hangleiter_appendix} is stated in terms of the \(\ell_1\)-distance threshold \(\rho\), whereas Sec.~\ref{sec:certification-bounds} uses total variation,
\begin{equation}
d_{\mathrm{TV}}(P,Q)
=
\frac12
\|P-Q\|_1.
\end{equation}
Hence,
$\rho
=
2\tau,
$
where \(\tau\) is the total-variation certification threshold. Substituting Eq.~\eqref{eq:minentropy_appendix} into Theorem~\ref{thm:hangleiter_appendix}, and replacing \(\rho=2\tau\), gives
\begin{equation}
m_{\mathrm{cert}}(Q,\tau)
\in
\Omega
\left[
\frac{2^{h/2}}{\tau^2}
\left(
1-4\tau-2^{-h}
\right)^{3/2}
\right],
\end{equation}
where universal constants have been absorbed into the \(\Omega(\cdot)\) notation, which is precisely Corollary~\ref{cor:min-entropy-certification}.

\end{document}